\documentclass[letterpaper]{article} 
\usepackage{aaai23}  
\usepackage{times}  
\usepackage{helvet}  
\usepackage{courier}  
\usepackage[hyphens]{url}  
\usepackage{graphicx} 
\usepackage{natbib}  
\usepackage{caption} 
\usepackage{booktabs}
\usepackage{algorithm}
\usepackage{algcompatible}
\usepackage{bm}
\usepackage{amssymb}
\usepackage{amsmath}
\usepackage{amsthm}
\usepackage{subfig}
\usepackage{makecell}
\usepackage{multirow}
\usepackage{url}
\usepackage{color}

\algloopdefx{RETURN}[1][]{\textbf{return} #1}

\newtheorem{theorem}{Theorem}
\newtheorem{lemma}{Lemma}

\theoremstyle{definition}
\newtheorem{definition}{Definition}

\usepackage{newfloat}
\usepackage{listings}
\DeclareCaptionStyle{ruled}{labelfont=normalfont,labelsep=colon,strut=off} 
\floatstyle{ruled}
\newfloat{listing}{tb}{lst}{}
\floatname{listing}{Listing}
\title{Submodular Maximization under the Intersection of \\ Matroid and Knapsack Constraints}
\author{
    Yu-Ran Gu,
    Chao Bian,
    Chao Qian
}
\affiliations{
    \textsuperscript{\rm}State Key Laboratory for Novel Software Technology, Nanjing University, Nanjing 210023, China\\

    \{guyr, bianc, qianc\}@lamda.nju.edu.cn
}

\usepackage{bibentry}

\begin{document}

\maketitle

\begin{abstract}
Submodular maximization arises in many applications, and has attracted a lot of research attentions from various areas such as artificial intelligence, finance and operations research. Previous studies mainly consider only one kind of constraint, while many real-world problems often involve several constraints. In this paper, we consider the problem of submodular maximization under the intersection of two commonly used constraints, i.e., $k$-matroid constraint and $m$-knapsack constraint, and propose a new algorithm SPROUT by incorporating partial enumeration into the simultaneous greedy framework. We prove that SPROUT can achieve a polynomial-time approximation guarantee better than the state-of-the-art algorithms. Then, we introduce the random enumeration and smooth techniques into SPROUT to improve its efficiency, resulting in the SPROUT++ algorithm, which can keep a similar approximation guarantee. Experiments on the applications of movie recommendation and weighted max-cut demonstrate the superiority of SPROUT++ in practice.
\end{abstract}

\section{Introduction}
Submodular maximization, i.e., maximization of a set function which satisfies the diminishing returns property under some constraints, arises in many applications, e.g., influence maximization~\citep{kempe2003maximizing,qian2017constrained}, data summarization~\citep{lin2011class, sipos2012temporal, dasgupta2013summarization}, sparse regression~\citep{DBLP:conf/icml/DasK11, qian2015subset}, sensor placement~\citep{krause2012near}, adversarial attack~\cite{liu2021efficient}, and human assisted learning~\cite{de2020regression,liu2023hal}. This problem is generally NP-hard, and the design of polynomial-time approximation algorithms has attracted much attention.

The pioneering work of~\citet{nemhauser1978analysis} and~\citet{fisher1978analysis} revealed the good performance of the classic greedy algorithm which achieves an approximation ratio of $(1-1 / e)^{-1}$ under the cardinality constraint and a $(k+1)$-approximation under the more general $k$-matroid constraint when maximizing a monotone submodular function. When the objective function is non-monotone, \citet{lee2010submodular} provided a $(k+1+1 / (k-1)+\epsilon)$-approximation for the $k$-matroid by local search requiring $\text{poly}(n) \cdot \exp(k, \epsilon)$ running time, where $n$ is the problem size and $\epsilon >0$. These studies mainly focus on submodular maximization under one kind of constraint, while real-world problems often deal with multiple constraints simultaneously, e.g., movie recommendation under cardinality and rating constraints~\cite{pmlr-v48-mirzasoleiman16}, and Twitter text summarization under cardinality and date constraints~\citep{badanidiyuru2020submodular}. 

For submodular maximization under matroid and knapsack constraints,~\citet{chekuri2010dependent} proposed an algorithm which achieves a $(k / 0.19+\epsilon)$-approximation for $k$-matroid and $m$-knapsack constraints with the time complexity of $\text{poly}(n) \cdot \exp(k, m, \epsilon)$, which, however, may be unacceptable.~\citet{pmlr-v48-mirzasoleiman16} developed the first practical algorithm, FANTOM, offering a $((1+\epsilon)(1+1 / k)(2k+2m+1))$-approximation for the intersection of the very general $k$-system and $m$-knapsack constraints with $\tilde{O}(n^2 / \epsilon)$ oracle calls. For the sake of simplicity, we ignore poly-logarithmic factors by using the $\tilde{O}$ notation. More recently,~\citet{feldman2020you} designed D{\footnotesize ENSITY}S{\footnotesize EARCH}SGS based on the simultaneous greedy algorithmic framework to solve this problem, which achieves an approximation ratio of $(1+\epsilon)(k+2m)+O(\sqrt{k+m})$ with $\tilde{O}(n / \epsilon)$ oracle calls. The approximation ratio becomes $(1+\epsilon)(k+2m)+O(\sqrt{m})$ for the intersection of $k$-extendible and $m$-knapsack constraints and $(1+\epsilon)(k+2m+1)$ for a monotone objective function, where $k$-extendible is a subclass of $k$-system. 

\begin{table*}[htbp]
\centering

\begin{tabular}{c|c|c}

Algorithm &
Approximation &
Running Time \\

\hline

FANTOM~\cite{pmlr-v48-mirzasoleiman16}
 &
$(1 + \epsilon)(2k + (2 + 2/k)m)
+ O(1)$ &
$\tilde{O}(n^2 / \epsilon)$ \\

D{\footnotesize ENSITY}S{\footnotesize EARCH}SGS~\cite{feldman2020you} &

$(1 + \epsilon)(k + 2m)
+ O(\sqrt{m})$ &
$\tilde{O}(n / \epsilon)$ \\

SPROUT (this paper) &
$(1 + \epsilon)(k + m)
+ O(\sqrt{m})$ &
$\tilde{O}(n^2 / \epsilon)$ \\

\end{tabular}
\caption{Comparison of the state-of-the-art algorithms for submodular maximization under the intersection of $k$-matroid and $m$-knapsack constraints. For the proposed SPROUT algorithm, the parameter $C=1$ is used here.} 
\label{table1}
\end{table*}

Note that the constraints (i.e., $k$-system) considered in previous studies may be so general that the proposed algorithms may not perform well under some important subclasses of these constraints. In this paper, we consider the (not necessarily monotone) submodular maximization problem under the intersection of $k$-matroid and $m$-knapsack constraints, which arises in numerous applications, e.g.,
vertex cover~\citep{DBLP:conf/nips/DelongVOB12}, weighted max-cut~\citep{DBLP:conf/colt/FeldmanHK17, DBLP:conf/icml/Haba0FK20}, video summarization~\citep{DBLP:conf/cvpr/GygliGG15, DBLP:conf/nips/FeldmanK018}, image summarization and revenue maximization~\citep{pmlr-v48-mirzasoleiman16}.
We propose a Simultaneous and Partial enumeRation cOnstrained sUbmodular maximizaTion algorithm, called  SPROUT, by incorporating the partial enumeration technique into the simultaneous greedy algorithmic framework. SPROUT offers an opportunity of balancing the approximation guarantee with the running time by a parameter $C$, i.e., as $C$ increases, the approximation ratio improves while the time complexity raises.
In particular, when $C = 1$, SPROUT achieves a $((1+\epsilon)(k+m+3+2\sqrt{m+1}))$-approximation using $\tilde{O}(n^{2} / \epsilon)$ oracle calls,\footnote{As in~\cite{feldman2020you}, the dependence on $k$ and $m$ is suppressed from the running time.} which is better than the state-of-the-art algorithms, as shown in Table~\ref{table1}. When the objective function is monotone, the approximation ratio improves to $(1+\epsilon)(k+m+1)$. In this case, our algorithm is incomparable with Algorithm 5 in~\cite{DBLP:conf/nips/LiF0K22}, which achieves a $((1+O(\epsilon))(k + 7m/4 + 1))$-approximation for $k$-system and $m$-knapsack constraints in essentially linear time. Besides, the BARRIER-GREEDY++ algorithm in~\cite{badanidiyuru2020submodular} achieves a $(k+1+\epsilon)$-approximation for $k$-matchoid and $m$-knapsack constraints (where $m \leq k$) in $\tilde{O}(n^3)$ time, where the $k$-matchoid constraint is a generalization of the $k$-matroid constraint. 

Since the partial enumeration used in SPROUT may be too time-consuming and some good solutions in binary search are ignored, we propose a more practical and efficient algorithm, i.e., SPROUT++, by introducing the random enumeration and smooth techniques into SPROUT. We prove that SPROUT++ can achieve a similar approximation guarantee to SPROUT. Experiments are conducted on the applications of movie recommendation and weighted max-cut, demonstrating the superior performance of SPROUT++ in practice.

\section{Preliminaries}

Given a ground set $\mathcal{N}$, we study the submodular functions $f : 2^{\mathcal{N}} \rightarrow \mathbb{R}$, i.e., set functions satisfying the diminishing returns property. Specifically, a set function is submodular if $f(e|A) \geq f(e|B), \ \forall A \subseteq B \subseteq \mathcal{N} \text{ and } e \notin B$, where $f(S|A) \triangleq f(S\cup A) - f(A)$ means the marginal gain of adding a set $S$ to $A$. Note that we do not distinguish the element $e$ and the corresponding single-element set $\{e\}$ for convenience. Without loss of generality, we assume the functions are non-negative. 

Now we introduce the considered constraints. Given a set system $(\mathcal{N},\mathcal{I})$ where $\mathcal{I} \subseteq 2^\mathcal{N}$, $(\mathcal{N},\mathcal{I})$ is called an independence system if (1) $\emptyset \in \mathcal{I}$; (2) $\forall A \subseteq B \subseteq \mathcal{N}$, if $B \in \mathcal{I}$ then $A \in \mathcal{I}$.
An independence system is called a matroid $\mathcal{M}(\mathcal{N},\mathcal{I})$ if $\forall A, B \in \mathcal{I}$ and $|A| < |B|$, there is $e \in B \backslash A$ such that $A \cup e \in \mathcal{I}$. We present $k$-matroid in Definition~\ref{def:matroid}.

\begin{definition}[$k$-Matroid] \label{def:matroid}
Given $k$ matroids $\mathcal{M}_i ( \mathcal{N} , \mathcal{I}_i )$ defined on the ground set $\mathcal{N}$, $k$-matroid is a matroid $\mathcal{M}(\mathcal{N},\mathcal{I})$, where $\mathcal{I} = \bigcap_{i=1}^{k}{\mathcal{I}_i }$.
\end{definition}
Besides, given a modular cost function $c:2^\mathcal{N} \rightarrow \mathbb{R}$, i.e., $\forall A \subseteq \mathcal{N}, c(A)=\sum_{e \in A}c(e)$, and a budget $B$, the knapsack constraint means that the cost of a subset should be upper bounded by $B$, i.e., $c(S) \le B$. Without loss of generality, the budget $B$ is normalized to 1. Next, we present $m$-knapsack constraint in Definition~\ref{def:knapsack}. We use $[m]$ (where $m$ is a positive integer) to denote the set $\{1,2,\ldots,m\}$. 
\begin{definition}[$m$-Knapsack Constraint] \label{def:knapsack}
Given $m$ modular cost functions $c_1, c_2, \dots , c_m$, a set $S\subseteq \mathcal{N}$ satisfies the $m$-knapsack constraint if and only if $\forall i\in [m], c_i(S)\le 1$.
\end{definition}

Instead of meeting one of these constraints, real-world applications often involve them simultaneously. In this paper, we study the problem of submodular maximization under the intersection of $k$-matroid and $m$-knapsack constraints.

\begin{definition}[Submodular Maximization Under the Intersection of $k$-Matroid and $m$-Knapsack Constraints] \label{def:problem}
Given a submodular function $f : 2^{\mathcal{N}} \rightarrow \mathbb{R}$, a $m$-knapsack constraint with cost functions $c_1, c_2,\cdots, c_m$, and a $k$-matroid $\mathcal{M}(\mathcal{N},\bigcap_{i=1}^{k}{\mathcal{I}_i})$, to find $\mathop{\arg\max}_{S \subseteq \mathcal{N}} f(S)$ such that $S \in \bigcap_{i=1}^{k}{\mathcal{I}_i}$ and $\forall i \in [m], c_{i}(S) \leq 1$.
\end{definition}

\citet{pmlr-v48-mirzasoleiman16} proposed FANTOM which can find a good solution by employing density threshold greedy multiple times. If the solutions generated-so-far are not good, there must exist high-quality solutions in the elements that never being chosen before, and density threshold greedy will continue to search.~\citet{feldman2020you} provided D{\footnotesize ENSITY}S{\footnotesize EARCH}SGS achieving a better approximation guarantee with lower time complexity by utilizing the simultaneous greedy framework. However, it seems that D{\footnotesize ENSITY}S{\footnotesize EARCH}SGS may be so delicate that it achieves previously the best approximation guarantee while behaving not that well in practice, as mentioned in~\cite{feldman2020you}. Aimed to get around it, we propose the algorithms, SPROUT and SPROUT++.

\section{SPROUT Algorithm}

\begin{algorithm}[htb]
    \begin{algorithmic}[1]
    \REQUIRE{Objective function $f:2^{\mathcal{N}}\rightarrow \mathbb{R}_{+}$, $k$ matroids $\mathcal{M}_i(\mathcal{N},\mathcal{I}_i)$ and $m$ cost functions $c_i:\mathcal{N} \rightarrow \mathbb{R}_{+}$}
        
    \algnewcommand\algorithmicparam{\textbf{Parameter:}}
    \algnewcommand\Param{\item[\algorithmicparam]}
    \Param{Error params $\delta, \epsilon$, correction params $\beta, \gamma$, enumeration param $C$ and number $\ell$ of solutions}
    
    \ENSURE{A set $S$ s.t. $S \in \bigcap_{i=1}^{k} \mathcal{I}_i$ and $\forall i\in [m], c_i(S) \leq 1 $}
    
    \FOR{each feasible $\mathcal{A} \subseteq \mathcal{N}$ with $C$ elements}
    \STATE $z_{\mathcal{A}}(S) \triangleq f(S|\mathcal{A})$.

    \STATE $\mathcal{N}' \triangleq \{e \in \mathcal{N} | e \notin \mathcal{A} \wedge C \cdot z_{\mathcal{A}}(e) \leq f(\mathcal{A}) \}$.

    \STATE $\mathcal{M}'_i(\mathcal{N}',\mathcal{I}'_i) \triangleq $ contraction of $\mathcal{M}_i(\mathcal{N},\mathcal{I}_i)$ by $\mathcal{A}$.
    
    \STATE $\mathcal{I}' \triangleq \bigcap_{i=1}^{k}\mathcal{I}'_i$.
    
    \STATE Decrease knapsack budgets by $c_i(\mathcal{A})$ and normalize \par \hskip -1em each of them to 1.

    \STATE $\text{Let } S_0 = \emptyset$, and $\mathcal{V}$ be the maximum $z_{\mathcal{A}}$ value of\par \hskip -1em a single feasible element in $\mathcal{N}'$.
    
    \STATE $\text{Let } b_{1}=1 \text{ and } b_{0}=\left \lceil \log{|\mathcal{N}'|} / \delta \right \rceil$.
    
    \WHILE{$| b_{1} - b_{0} | > 1$}
    
    \STATE $\rho = \beta \mathcal{V} (1+\delta)^{\left \lfloor (b_{1} + b_{0} + 1) / 2 \right \rfloor} + \gamma f(\mathcal{A}) / C$.
    
    \STATE $S_K\! = \!\text{K}\text{\footnotesize{NAPSACK}}\text{SGS}(z_{\mathcal{A}},\!\mathcal{N}',\!\mathcal{I}',\!\{c_i\}_{i=1}^m,\!\ell,\!\rho,\!\epsilon)$.

    \STATE Add $S_K$ to $S_0$.
    
    \STATE $b_{E} = \left \lfloor (b_{1} + b_{0} + 1) / 2 \right \rfloor$.
    
    \ENDWHILE
    
    \STATE $S_{\mathcal{A}} = \mathop{\arg\max}_{S \in S_0} f(S)$.
     
    \ENDFOR

    \STATE $\mathcal{A}^* = \mathop{\arg\max}_{\mathcal{A}} f(\mathcal{A} \cup S_{\mathcal{A}})$ over all feasible $\mathcal{A} \subseteq \mathcal{N}$.
    
    \RETURN{$\mathcal{A}^* \cup S_{\mathcal{A}^*}$}
    
    \caption{SPROUT}
    \label{algo:SPROUT}
\end{algorithmic}
\end{algorithm}
The basic idea of SPROUT is incorporating a partial enumeration technique~\cite{badanidiyuru2020submodular} into the simultaneous greedy framework~\cite{feldman2020you}, such that it can be more robust in practice and provide a better approximation ratio.
Specifically, SPROUT as presented in Algorithm~\ref{algo:SPROUT} enumerates the feasible set $\mathcal{A}$ with a size of $C$ as the first part of solution set, and then selects the set $S_{\mathcal{A}}$ from the remaining elements using the K{\footnotesize NAPSACK}SGS subroutine as presented in Algorithm~\ref{algo:KNAPSACKSGS}. Finally, it returns the best solution $\mathcal{A} \cup S_{\mathcal{A}}$ over all feasible $\mathcal{A}$. 
\begin{algorithm}[htb]
    \begin{algorithmic}[1]
    \REQUIRE{Objective function $z:2^{\mathcal{N}'}\rightarrow \mathbb{R}_{+}$, $k$ matroids $\mathcal{M}'_i(\mathcal{N}',\mathcal{I}'_i)$ and $m$ cost functions $c'_i:\mathcal{N}' \rightarrow \mathbb{R}_{+}$}
    
    \algnewcommand\algorithmicparam{\textbf{Parameter:}}
    \algnewcommand\Param{\item[\algorithmicparam]}
    \Param{Error param $\epsilon$, number $\ell$ of solutions, density ratio $\rho$}
    
    \ENSURE{A set $S$ s.t. $S \in \bigcap_{i=1}^{k} \mathcal{I}'_i$ and $\forall  i \in [m], c'_i(S) \leq 1 $}

    \STATE Let $S_{0}^{i} = \emptyset$ for all $i \in [\ell]$ and $\mathcal{N}'_0 = \mathcal{N}'$.
    
    \STATE Let $t = 1$ and $\tau = \mathcal{V}$.
    
    \WHILE{$\tau > \epsilon \mathcal{V} / n$}
    \FOR{each pair $(a,i)$ with $a \in \mathcal{N}'_{t-1}, i \in [\ell]$}
    
    \IF{$S_{t-1}^{i} \cup a \in \mathcal{I}'$ and $z(a|S_{t-1}^{i}) \geq \max \{\tau, \rho $ \par
        \hskip\algorithmicindent $\cdot \sum_{j=1}^{m}c_{j}'(a)\}$}
    
    \IF{$c_{j}'(a) \leq 1, \forall j \in [m]$}
    
    \STATE Let $a_{t} = a$ and $i_{t} = i$.
    
    \STATE For $i' \in [\ell] $, 
    $
    S_{t}^{i'} = 
    \left\{
    \begin{aligned}
    S_{t-1}^{i_{t}} \cup a_{t}, \text{if } i' = i_{t}\\
    S_{t-1}^{i}, \text{if } i' \neq i_{t}
    \end{aligned}
    \right.
    $.
    
    \STATE $\mathcal{N}'_{t} = \mathcal{N}'_{t-1} \backslash a_{t}$.
    
    \STATE $t = t + 1$.
    
    \ENDIF
    \ENDIF
    \ENDFOR
    
    \STATE $\tau = (1 - \epsilon) \tau$.
    
    \ENDWHILE
    

    \RETURN{$\mathop{\arg\max}_{S \in \{S_{t-1}^{i}\}_{i=1}^{\ell}} z(S)$}
    
    \caption{K{\footnotesize NAPSACK}SGS: Subroutine of SPROUT \\~\citep{feldman2020you}}
    \label{algo:KNAPSACKSGS}
    \end{algorithmic}
\end{algorithm}
In the following, we describe in more details about SPROUT. 
First, we define a new objective function $z_{\mathcal{A}}$ in line~2 for the enumerated feasible set $\mathcal{A} \subseteq \mathcal{N}$ in line~1. Then, a reduced ground set $\mathcal{N}'$ is derived by removing elements belonging to $\mathcal{A}$ or whose value of $z_{\mathcal{A}}$ is larger than $f(\mathcal{A}) / C$ from $\mathcal{N}$ in line~3. After selecting the set $\mathcal{A}$, we need to contract matroids from $\mathcal{M}_i(\mathcal{N},\mathcal{I}_i)$ to $\mathcal{M}'_i(\mathcal{N}',\mathcal{I}'_i)$  by $\mathcal{A}$ in line~4, where $\mathcal{I}'_i$ is the set $\{S \subseteq \mathcal{N}': S \cup \mathcal{A} \in \mathcal{I}_i \}$ based on the concept of matroid contraction~\cite{white1986theory}. For knapsack constraints, the budget for each cost function is decreased by the corresponding cost of $\mathcal{A}$ in line~6. 

We start with the subroutine, i.e., K{\footnotesize NAPSACK}SGS~\cite{feldman2020you} as presented in Algorithm~\ref{algo:KNAPSACKSGS}, which utilizes the density threshold in its simultaneous greedy framework. At a high level, K{\footnotesize NAPSACK}SGS simultaneously maintains $\ell$ disjoint candidate solution sets and inserts an element to one of them at a time by a selection criterion utilizing the density threshold. 
More concretely, it adds an element to one candidate solution set $S$ when its density, i.e., the marginal gain divided by the sum of the knapsack costs, is not smaller than the density ratio $\rho$ in line~5.
Moreover, K{\footnotesize NAPSACK}SGS decreases threshold $\tau$ in line~14 to limit the number of iterations. 

Since the density ratio is key to SPROUT, it is essential to derive an appropriate value. Inspired by~\cite{feldman2020you}, SPROUT uses binary search in lines~8--14 to approximate the best density ratio $\rho^*$, where $\beta$ is used to ensure that $\rho^*$ is included in the range of search. In this procedure, each $S_K$ generated by K{\footnotesize NAPSACK}SGS is added to $S_0$, and $E$ in $b_E$ indicates whether the knapsack constraints are violated in line~6 of K{\footnotesize NAPSACK}SGS during the execution, i.e., $E = 0$ if the subroutine never violates them and $E = 1$ otherwise. 
For each $\mathcal{A}$, $S_\mathcal{A} $ with the maximal $f$ value is selected from $S_0$ in line~15. Finally, the union of $\mathcal{A}$ and its corresponding $S_\mathcal{A}$ maximizing $f$ is returned.

\subsection{Theoretical Analysis}

We prove in Theorem~\ref{theo:1} that SPROUT can achieve an approximation ratio of $(1+\epsilon)(k+m+3+2\sqrt{m+1})$ using $\tilde{O}(n^{2} / \epsilon)$ oracle calls, and this ratio improves to $(1+\epsilon)(k+m+1)$ when the objective function is monotone. 

The time complexity is mainly measured by the number of oracle calls and arithmetic operations since we evaluate the objective function and the constraints through a value and a membership oracle, respectively. Because the computational cost of arithmetic operations is much less than that of oracle calls in most applications, we focus more on the number of oracle calls. For the ease of presentation, we use $S_\text{OPT}$ and $\text{OPT}$ to denote an optimal solution and its objective value of the problem in Deﬁnition~\ref{def:problem}, respectively. Besides, the solution $S_{\mathcal{A}}\cup \mathcal{A}$ obtained by SPROUT is represented by $\mathcal{S}$. 

The proof of Theorem~\ref{theo:1} relies on Lemma~\ref{lem:1}, which shows the approximation of solutions obtained in the binary search procedure of SPROUT. Our proof is mainly attributed to the analysis in~\cite{feldman2020you}.

\begin{lemma} \label{lem:1}
In SPROUT, $f(\mathcal{A} \cup S_K) \geq \min \{\rho + (1- 1 / C)f(\mathcal{A}), \frac{1-\epsilon}{p+1}( (1- 1 / \ell -\epsilon )z_{\mathcal{A}}(S_{\text{OPT}}') - \rho m ) + f(\mathcal{A}) \}$ for each generated $\rho$ in line~10 and corresponding $S_K$, where $S_{\text{OPT}}'$ refers to an optimal solution for the reduced instance, and $p= \max\{\ell-1, k\}$.
\end{lemma}
\begin{proof}
The following analysis includes two parts based on the value of the indicator $E$. When $E=1$, there is a candidate solution set $S_{t}^{i}$ and an element $a$ satisfying that $S_{t}^{i}\cup a$ obeys the $k$-matroid constraint and violates the $m$-knapsack constraint, i.e., $S_{t}^{i} \cup a \in \mathcal{I}'$ and $c'_{j}(S_{t}^{i} \cup a) > 1$ for some $j \in [m]$. We use $A$ to represent $S_{t}^{i} \cup a$ and show that $z_{\mathcal{A}}(A) > \rho$, despite the infeasibility of $A$. Considering Corollary~54 in~\cite{feldman2020you}, we permutate the elements of $A$ in the order of being added to $S_{t}^{i}$, which means $A = \{a_i\}_{i=1}^{k}$ and $a_k$ is $a$. Let $A_{i} = \{a_j\}_{j=1}^{i}$ for $i \in [k]$ and $A_0 = \emptyset$. As a result, we derive that $z_{\mathcal{A}}(A) = \sum_{i=1}^{k}z_{\mathcal{A}}(a_{i}|A_{i-1}) \geq \sum_{i=1}^{k}\rho\sum_{j=1}^{m}c'_{j}(a_i)=\rho\sum_{j=1}^{m}c'_{j}(A) > \rho$, since all of the elements in $A$ are added with a density ratio no less than $\rho$, and $A$ violates the $m$-knapsack constraint. By line~16 of Algorithm~\ref{algo:KNAPSACKSGS}, we know that $z_{\mathcal{A}}(S_K)=\max_{S \in \{S_{T}^{i}\}_{i=1}^{\ell}} z_{\mathcal{A}}(S)$, where $T$ is the number of iterations executed in K{\footnotesize NAPSACK}SGS. By line~3 of Algorithm~\ref{algo:SPROUT}, we have $ C \cdot z_{\mathcal{A}}(a) \leq f(\mathcal{A})$. Thus, we obtain $f(\mathcal{A} \cup S_K) = f(\mathcal{A}) + z_{\mathcal{A}}(S_K) \geq f(\mathcal{A}) + z_{\mathcal{A}}(S_{t}^{i}) = f(\mathcal{A}) + z_{\mathcal{A}}(A \setminus a) \geq f(\mathcal{A}) + z_{\mathcal{A}}(A) - z_{\mathcal{A}}(a) \geq \rho + (1- 1 / C)f(\mathcal{A})$, where the second inequality is by the submodularity of $z_{\mathcal{A}}$.

When $E = 0$, we apply Corollary~54 in~\cite{feldman2020you} to K{\footnotesize NAPSACK}SGS, which yields $
z_{\mathcal{A}}(S_K) = \max_{S \in \{S_{T}^{i}\}_{i=1}^{\ell}} z_{\mathcal{A}}(S) \geq \sum_{i=1}^{\ell}z_{\mathcal{A}}(S_{T}^{i}) / \ell \geq \frac{1-\epsilon}{p+1} \left( \sum_{i=1}^{\ell}z_{\mathcal{A}}(S_{\text{OPT}}'\cup S_{T}^{i}) / \ell - \epsilon \mathcal{V} - \rho m  \right)$, where the set $S_{\text{OPT}}'$ refers to an optimal solution for the reduced instance, and $p = \max\{\ell-1, k\}$ as stated by Proposition 8 in~\citep{feldman2020you}. Besides, if $z_{\mathcal{A}}$ is monotone, this conclusion can improve to 
\begin{equation} \label{eq:mono}
z_{\mathcal{A}}(S_K) \geq (1-\epsilon) \left( z_{\mathcal{A}}(S_{\text{OPT}}') - \epsilon \mathcal{V} - \rho m  \right)/(p+1).
\end{equation}
Let $S_u$ be a set selected from the disjoint sets $\{S_{T}^{i}\}_{i=1}^{\ell}$ uniformly at random. Obviously, each element in $\mathcal{N}'$ appears in $S_u$ with probability at most $1 / \ell$. We have $\sum_{i=1}^{\ell}z_{\mathcal{A}}(S_{\text{OPT}}'\!\cup \!S_{T}^{i}) / \ell \!=\! \mathbb{E}[z_{\mathcal{A}}(S_{\text{OPT}}'\!\cup\! S_u)] \!\geq\! (1\!-\! 1 / \ell)z_{\mathcal{A}}(S_{\text{OPT}}')$, where the inequality holds by Lemma~2.2 in~\cite{buchbinder2014submodular}. Thus, $f(\mathcal{A} \!\cup\! S_K) \!= \!  f(\mathcal{A}) \!+\! z_{\mathcal{A}}(S_K) \!\geq\! \frac{1-\epsilon}{p+1} \left( (1- 1 / \ell -\epsilon )z_{\mathcal{A}}(S_{\text{OPT}}') - \rho m \right) + f(\mathcal{A})$. 
\end{proof}

\begin{theorem} \label{theo:1}
For the problem in Definition~\ref{def:problem}, when the error parameters $\delta=\epsilon$, and the number $\ell$ of solutions is $P+1$, SPROUT achieves an approximation ratio of roughly\footnote{The precise approximation ratio is shown in Eq.~\eqref{eq:summary}.} $\left(\frac{1-\epsilon}{k+m+3+2\sqrt{m+1}}+\frac{(1-\epsilon)C}{r}\right)^{-1}$ using $\tilde{O}(Pn^{C+1} / \epsilon)$ oracle calls and $\tilde{O}(Pmn^{C+1} / \epsilon)$ arithmetic operations, where $P = \max\{\lceil \sqrt{1+m} \rceil, k\}$ and $r$ is the size of $S_{\text{OPT}}$. 

When $C = 1$, the approximation ratio becomes $(1+\epsilon)(k+m+3+2\sqrt{m+1})$, and the number of oracle calls becomes $\tilde{O}(Pn^2 / \epsilon)$. If the objective function $f$ is monotone, the approximation ratio improves to $(1+\epsilon)(k+m+1)$.
\end{theorem}
\begin{proof}
SPROUT enumerates each feasible set $\mathcal{A}$ of size $C$ in line~1, thus it can find a set $\mathcal{A}$ with max-value elements in the optimal solution. In the following, we are to show that the subroutine K{\footnotesize NAPSACK}SGS can find a corresponding set $S_{\mathcal{A}}$ such that $S_{\mathcal{A}} \cup \mathcal{A}$ is good enough. 

In particular, we consider the case that $\mathcal{A}$ is the subset of the optimal solution $S_{\text{OPT}}$, containing the first $C$ elements of $S_{\text{OPT}}$ in the greedy ordering with respect to the objective function $f$, inspired by~\cite{badanidiyuru2020submodular}. All elements of $S_{\text{OPT}} \backslash \mathcal{A}$ are kept in the reduced ground set, because they have a marginal value of at most $f(\mathcal{A}) / C$ due to the greedy characteristic during its generation. Thus, $S_{\text{OPT}}\backslash \mathcal{A}$ is feasible for the reduced instance. 

In terms of Lemma~\ref{lem:1}, we derive $f(\mathcal{S}) \geq \min \{\rho + (1- 1 / C)f(\mathcal{A}), \frac{1-\epsilon}{p+1}\left( (1- 1 / \ell -\epsilon )z_{\mathcal{A}}(S_{\text{OPT}}') - \rho m \right) + f(\mathcal{A}) \}$. To maximize this lower bound, the ideal density ratio is 
$$
\rho^* = \frac{(1-\epsilon)(1- 1 / \ell -\epsilon)z_{\mathcal{A}}(S_{\text{OPT}}') + (p+1)f(\mathcal{A}) / C}{p + 1 + m(1-\epsilon)},
$$
and the corresponding value of the lower bound is 
\begin{equation*}
    \begin{aligned}
        f^*(\mathcal{S}) = & \ \frac{(1-\epsilon)(1- 1 / \ell -\epsilon)}{p+1+(1-\epsilon)m}z_{\mathcal{A}}(S_{\text{OPT}}') \\
        &+ \left(1-\frac{(1-\epsilon)m}{C\left( p+1+(1-\epsilon)m \right)}\right)f(\mathcal{A}).        
    \end{aligned}
\end{equation*}
However, $\rho^*$ is hard to be derived in practice because it contains the item $z_{\mathcal{A}}(S_{\text{OPT}}')$. Next, based on~\cite{feldman2020you}, we will show that SPROUT can use binary search in lines~8--14 to obtain a good estimation of $\rho^*$. 

SPROUT sets $\rho = \beta \mathcal{V} (1+\delta)^{\left \lfloor (b_{1} + b_{0} + 1) / 2 \right \rfloor} + \gamma f(\mathcal{A}) / C$, as shown in line~10 of Algorithm~\ref{algo:SPROUT}. By setting the correction parameter $\beta = (1-\epsilon)(1- 1 / \ell -\epsilon) / (p + 1 + m(1-\epsilon))$ and $\gamma = (p+1) / (p+1+m(1-\epsilon))$, we have 
$$
\rho = 
\frac{(1-\epsilon)(1- 1 / \ell -\epsilon)\mathcal{V}(1+\delta)^{b} + (p+1)f(\mathcal{A}) / C}{p + 1 + m(1-\epsilon)},
$$
where $b \!=\! \left \lfloor (b_{1} \!+\! b_{0} \!+ \!1) / 2 \right \rfloor$. Considering the submodularity of $z_{\mathcal{A}}$, we have $\mathcal{V} \!\leq\! z_{\mathcal{A}}(S_{\text{OPT}}') \!\leq\! |\mathcal{N}'|\mathcal{V}$. Thus, we set $b_{1}\!=\!1$ and $b_{0}\!=\!\left \lceil \log{|\mathcal{N}'|} / \delta \right \rceil$, to make the initial search range of $\rho$ contain $\rho^*$. We will show that during the binary search procedure, SPROUT can either get a solution with good approximation guarantee directly or maintain the best density ratio $\rho^*$ in the range of search (i.e., the range of $\rho$ where $b \!\in\! [b_1, b_0]$). We consider three cases for $\rho$ and $E$ as follows.
\begin{itemize}
    \item [(1)] There exists one iteration of binary search, such that $\rho \leq \rho^*$ and $E=0$. According to the analysis of the case $E=0$ in the proof of Lemma~\ref{lem:1}, we know that SPROUT has got a solution $\mathcal{S}$ satisfying $f(\mathcal{S}) \geq \frac{1-\epsilon}{p+1}\left( (1- 1 / \ell -\epsilon )z_{\mathcal{A}}(S_{\text{OPT}}') - \rho m \right) + f(\mathcal{A}) \geq \frac{1-\epsilon}{p+1}\left( (1- 1 / \ell -\epsilon )z_{\mathcal{A}}(S_{\text{OPT}}') - \rho^* m \right) + f(\mathcal{A}) = f^{*}(\mathcal{S})$, where the second inequality holds by $\rho \leq \rho^*$. 
    \item [(2)] There exists one iteration of binary search, such that $\rho \geq \rho^*$ and $E=1$. According to the analysis of the case $E=1$ in the proof of Lemma~\ref{lem:1}, we know that SPROUT has got a solution $\mathcal{S}$ satisfying $f(\mathcal{S}) \geq \rho + (1- 1 / C)f(\mathcal{A}) \geq \rho^* + (1- 1 / C)f(\mathcal{A})= f^{*}(\mathcal{S})$, where the second inequality holds by $\rho \geq \rho^*$.
    \item [(3)] If the above cases have not occurred, then $\rho \leq \rho^*$ implies $E=1$, which will increase $b_1$ to $\lfloor(b_1+b_0+1)/2\rfloor$; and $\rho \geq \rho^*$ implies $E=0$, which will decrease $b_0$ to $\lfloor(b_1+b_0+1)/2\rfloor$. Thus, $\rho^*$ is always contained in the range of binary search (i.e., the range of $\rho$ where $b \in [b_1, b_0]$), and the final $\rho$ found by SPROUT satisfies $(1\!-\!\delta)\rho^* \!\leq\! \rho \!\leq \!\rho^*$.
\end{itemize}
Combining the analyses of the above three cases, we have
\begin{align*}
    f(\mathcal{S}) \geq & \ \frac{(1-\delta)(1-\epsilon)(1- 1 / \ell -\epsilon)}{p+1+m(1-\epsilon)}z_{\mathcal{A}}(S_{\text{OPT}}') \\
    &\quad+ \left(1-\frac{\delta(p+1)+m(1-\epsilon)}{C\left( p+1+m(1-\epsilon) \right)}\right)f(\mathcal{A}).    
\end{align*}

Because $S_{\text{OPT}}\backslash {\mathcal{A}}$ is feasible for the reduced instance and $S_{\text{OPT}}'$ is an optimal one, we have $z_{\mathcal{A}}(S_{\text{OPT}}') \geq z_{\mathcal{A}}(S_{\text{OPT}}\backslash {\mathcal{A}}) =  \text{OPT}-f(\mathcal{A})$. Furthermore, the greedy choice of $\mathcal{A}$ implies $f(\mathcal{A}) \geq C \cdot \text{OPT} / |S_{\text{OPT}}|=C \cdot \text{OPT} / r$. Since $\delta = \epsilon$ and $1 - 1 / \ell - \epsilon \geq (1-1/\ell)(1-2\epsilon)$, we have
\begin{align*}
    &f(\mathcal{S}) \geq \frac{(1-\epsilon)^{2}(1-2\epsilon)(1-1/\ell)}{p+1+m}\cdot \text{OPT}  \\
    &\quad+\left(\frac{C}{r}-\frac{\epsilon(p+1)+m(1-\epsilon)+C(1-\epsilon)^{3}}{r\left( p+1+m(1-\epsilon) \right)}\right)\cdot \text{OPT}.
\end{align*}
Let $Q = (1-1/\ell)/(p+1+m)$, and we have $p=P$ by setting $\ell=P + 1$. If $k > \lceil \sqrt{1+m} \rceil$, $Q = (k+m+2+\frac{m+1}{k})^{-1} \geq (k+m+2+\sqrt{m+1})^{-1}$. Otherwise, $Q = (m+2+\lceil \sqrt{1+m} \rceil+\frac{m+1}{\lceil \sqrt{1+m} \rceil})^{-1} \geq (m+3+2\sqrt{m+1})^{-1}$. Thus, SPROUT will return a solution $\mathcal{S}$ satisfying
\begin{align}\label{eq:summary}
    &f(\mathcal{S}) \geq {\bigg (}\frac{(1-\epsilon)^{2}(1-2\epsilon)}{k+m+3+2\sqrt{m+1}} \\
    &+\!\frac{(C\!-\!\epsilon)(P+\!1)\! + \!(C-\!1)m(1-\!\epsilon) \!- \!C(1-\!\epsilon)^{3}}{r\left( P+1+m(1-\epsilon) \right)}{\bigg)} \text{OPT},\nonumber
\end{align}
i.e., a roughly $\left(\frac{1-\epsilon}{k+m+3+2\sqrt{m+1}}\!+\!\frac{(1-\epsilon)C}{r}\right)^{-1}$ approximation ratio. When $C = 1$, the approximation ratio becomes $(1+\epsilon)(k+m+3+2\sqrt{m+1})$. Moreover, when $f$ is monotone, SPROUT can provide a $(1-\epsilon)^{-3}(k+m+1)$-approximation solution according to Eq.~\eqref{eq:mono} by setting $\ell \leq k + 1$, since the monotonicity of $f$ leads to the monotonicity of $z_{\mathcal{A}}$. As the factor $(1-\epsilon)^{3}$ is in the order of $1 + O(\epsilon)$, this ratio can be rewritten as $(1+\epsilon)(k+m+1)$.

We then analyze the time complexity of SPROUT. It is easy to find that SPROUT has at most $n^C$ iterations. When $\mathcal{A}$ is fixed, we first need to reduce the ground set which costs $O(n)$ oracle calls and arithmetic operations, and the most time-consuming procedure is the calls to the subroutine K{\footnotesize NAPSACK}SGS, each of which requires $\tilde{O}(\ell |\mathcal{N'}| / \epsilon)$ oracle calls and $\tilde{O}(m\ell |\mathcal{N}'| / \epsilon)$ arithmetic operations according to Observation 22 in~\citep{feldman2020you}. As the binary search method is employed, the number of calls is $O(\log{ (b_0/b_1)})=O(\log{(\log{ |\mathcal{N'}| / \delta)}}) = \tilde{O}(1)$. Thus, we can conclude that SPROUT requires $\tilde{O}(Pn^{C+1} / \epsilon)$ oracle calls and $\tilde{O}(Pmn^{C+1} / \epsilon)$ arithmetic operations.
\end{proof}

\section{SPROUT++ Algorithm}

Though SPROUT can achieve the best approximation guarantee, the exhaustive enumeration may be too time-consuming. Thus, we propose SPROUT++, an accelerated version of SPROUT. In a nutshell, SPROUT++ improves the efficiency by only randomly enumerating good single elements. It also introduces a smooth technique into the binary search procedure to make the algorithm more robust.

As presented in Algorithm~\ref{algo:SPROUT++}, SPROUT++ specifies $C$ as $1$ and randomly picks $t_c$ feasible sets (i.e., feasible single elements) in line~3 instead of enumerating over all possible $\mathcal{A}$. On top of that, SPROUT++ omits the elements with low value in line~4, and does not delete extra elements from the ground set in line~6. Besides, it employs a smooth technique in line~17 in the search procedure, since the search range shrinks so fast that many density ratios corresponding to good solutions may be missed. Technically, we use a parameter $\mu$ to decide the search range in the next iteration instead of sharply reducing it by half.

\begin{algorithm}[t!]
    \begin{algorithmic}[1]
    \REQUIRE{Objective function $f:2^{\mathcal{N}}\rightarrow \mathbb{R}_{+}$, $k$ matroids $\mathcal{M}_i(\mathcal{N},\mathcal{I}_i)$ and $m$ cost functions $c_i:\mathcal{N} \rightarrow \mathbb{R}_{+}$}
        
    \algnewcommand\algorithmicparam{\textbf{Parameter:}}
    \algnewcommand\Param{\item[\algorithmicparam]}
    \Param{Error params $\delta, \epsilon$, correction params $\beta, \gamma$, acceleration param $\alpha$, smooth param $\mu$, counter $t_c$ and number $\ell$ of solutions}
    
    \ENSURE{A set $S$ s.t. $S \in \bigcap_{i=1}^{k} \mathcal{I}_i$ and $\forall i\in [m]$, $c_i(S) \leq 1$}

    \STATE Let $e^*$ be the feasible element $e \in \mathcal{N}$ maximizing $f(e)$.

    \WHILE{$t_c > 0$}

    \STATE Randomly select a feasible single-element set $\mathcal{A} \subseteq \mathcal{N}$ \par \hskip -1em never being chosen before.

    \IF{$f(\mathcal{A}) \geq (1 - \alpha) f(e^*)$}

    \STATE $z_{\mathcal{A}}(S) \triangleq f(S|\mathcal{A})$.

    \STATE $\mathcal{N}' \triangleq \{e \in \mathcal{N} | e \notin \mathcal{A} \}$.

    \STATE $\mathcal{M}'_i(\mathcal{N}',\mathcal{I}'_i) \triangleq $ contraction of $\mathcal{M}_i(\mathcal{N},\mathcal{I}_i)$ by $\mathcal{A}$.
    
    \STATE $\mathcal{I}' \triangleq \bigcap_{i=1}^{k}\mathcal{I}'_i$.
    
    \STATE Decrease knapsack budgets by $c_i(\mathcal{A})$ and \par normalize each of them to 1.

    \STATE $\text{Let } S_0 = \emptyset$, and $\mathcal{V}$ be the maximum $z_{\mathcal{A}}$ value of \par \hskip 0em a single feasible element in $\mathcal{N}'$.
    
    \STATE $\text{Let } b_{1}=1 \text{ and } b_{0}=\left \lceil \log{|\mathcal{N}'|} / \delta \right \rceil$.
    
    \WHILE{$| b_{1} - b_{0} | > 1$}

    \STATE $b = \left \lfloor (b_{1} + b_{0} + 1) / 2 \right \rfloor$.
    
    \STATE $\rho = \beta \mathcal{V} (1+\delta)^{b} + \gamma f(\mathcal{A})$.
    
    \STATE $S_{K}\! = \!\text{K}\text{\footnotesize{NAPSACK}}\text{SGS}(z_{\mathcal{A}},\!\mathcal{N}'\!,\!\mathcal{I}'\!,\!\{c_i\}_{i=1}^m,\!\ell,\!\rho,\!\epsilon)$.

    \STATE Add $S_K$ to $S_0$.
    
    \STATE $b_{E} = b + (1 - 2E)(1 - 1 / \mu) | b_{E} - b|$.
    
    \ENDWHILE
    
    \STATE $S_{\mathcal{A}} = \mathop{\arg\max}_{S \in S_0} f(S)$.

    \STATE $t_c = t_c - 1$.

    \ENDIF
    
    \ENDWHILE

    \STATE $\mathcal{A}^* = \mathop{\arg\max}_{\mathcal{A}} f(\mathcal{A} \cup S_{\mathcal{A}})$ over all feasible $\mathcal{A} \subseteq \mathcal{N}$.
    
    \RETURN{$\mathcal{A}^* \cup S_{\mathcal{A}^*}$}
    
    \caption{SPROUT++}
    \label{algo:SPROUT++}
\end{algorithmic} 
\end{algorithm}

In Theorem~\ref{theo:2}, we prove that SPROUT++ can achieve a similar approximation guarantee to SPROUT with a high probability using much less time (depending on $t_c$) under an assumption in Eq.~\eqref{eq:assumption}. This assumption intuitively means that the objective value of each element in $S_{\text{OPT}}$ is relatively large, which can hold if the marginal gain of adding each element $e$ to $S_{\text{OPT}} \backslash e$ is large enough by the submodularity. It often appears in the problems of selecting small subsets from a relatively large ground set, where a small chosen subset may represent only part of the ground set and there are elements which can still contribute enough. Since SPROUT++ is for acceleration, such requirement naturally meets its large-scale application.

\begin{theorem} \label{theo:2}
For the problem in Definition~\ref{def:problem}, suppose that 
\begin{equation}\label{eq:assumption}
\forall a\in S_{\text{OPT}}, (1+\alpha)\cdot f(a) \geq f(e^*),
\end{equation}where $e^*$ is a feasible max-value element in $\mathcal{N}$ and $\alpha\! \leq \!(1\!-\!\epsilon)(p\!+\!1-\!(1\!-\!\epsilon)^{2}) / (\epsilon(p\!+\!1)\!+\!m(1\!-\!\epsilon))$. Then SPROUT++ offers an approximation ratio of $(1+\epsilon)(k+m+3+2\sqrt{m+1})$ with probability at least $1 - e^{-rt_c / n}$ using $\tilde{O}({\log^{-1}{(2\mu/(2\mu-1))}}t_cPn / \epsilon)$ oracle calls and $\tilde{O}({\log^{-1}{(2\mu/(2\mu-1))}}t_c Pmn / \epsilon)$ arithmetic operations, where $P = \max\{\lceil \sqrt{1+m} \rceil, k\}$, $r$ is the size of $S_\text{OPT}$, and $\mu$ is the smooth parameter.
\end{theorem}
\begin{proof}
Compared with SPROUT, we find that the shrinking ratio of the search range is now reduced to $(2\mu - 1) / 2\mu$ during the search procedure of SPROUT++, implying that the time complexity is multiplied by a factor of ${\log^{-1}{(2\mu/(2\mu-1))}}$. In the following, we use $\mathcal{A}_{1},\cdots,\mathcal{A}_{t_c}$ to denote the chosen $t_c$ feasible single-element sets.


Consider that there is $\mathcal{A} \in \{\mathcal{A}_{1},\cdots,\mathcal{A}_{t_c}\}$ such that $\mathcal{A} \subseteq S_{\text{OPT}}$. According to the proof of Lemma~\ref{lem:1}, when $E=1$, $f(\mathcal{S}) \geq f(\mathcal{A}) + z_{\mathcal{A}}(A) - z_{\mathcal{A}}(a) \geq f(\mathcal{A}) +\rho- z_{\mathcal{A}}(a) \geq \rho - \alpha f(\mathcal{A})$, where the last inequality holds by $z_{\mathcal{A}}(a) \leq f(e^*) \leq (1+\alpha)f(\mathcal{A})$ due to the assumption in Eq.~\eqref{eq:assumption}; when $E=0$, we have $f(\mathcal{S}) \geq\frac{1-\epsilon}{p+1}\left( (1- 1 / \ell -\epsilon )z_{\mathcal{A}}(S_{\text{OPT}}') - \rho m \right) + f(\mathcal{A})$. Thus, we can conclude that $f(\mathcal{S}) \geq \min \{\rho - \alpha f(\mathcal{A}), \frac{1-\epsilon}{p+1}\left( (1- 1 / \ell -\epsilon )z_{\mathcal{A}}(S_{\text{OPT}}') - \rho m \right) + f(\mathcal{A})\}$. By following the proof of Theorem~\ref{theo:1} with $C=1$ and 
$$
\rho^* = \frac{(1\!-\!\epsilon)(1\!-\!1/\ell\!-\!\epsilon)z_{\mathcal{A}}(S_{\text{OPT}}')+(1\!+\!\alpha)(p\!+\!1)f(\mathcal{A})}{p+1+m(1-\epsilon)},
$$
we can derive 
\begin{align*}
    &f(\mathcal{S}) \geq  \frac{(1-\delta)(1-\epsilon)(1- 1 / \ell -\epsilon)}{p+1+m(1-\epsilon)}z_{\mathcal{A}}(S_{\text{OPT}}') \\
    &\quad+ \frac{(1-\delta)(p+1)-\alpha\left(\delta(p+1)+m(1-\epsilon)\right)}{p+1+m(1-\epsilon)}f(\mathcal{A}).        
\end{align*}
As $\mathcal{A} \subseteq S_{\text{OPT}}$, we still have $z_{\mathcal{A}}(S_{\text{OPT}}') \geq z_{\mathcal{A}}(S_{\text{OPT}} \backslash \mathcal{A}) = \text{OPT} - f(\mathcal{A})$. Note that in the proof of Theorem~\ref{theo:1}, $\mathcal{A}$ is the best element in $S_{\text{OPT}}$, and thus $f(\mathcal{A}) \geq \text{OPT}/|S_{\text{OPT}}|=\text{OPT}/r$; while $\mathcal{A}$ here can be only guaranteed from $S_{\text{OPT}}$, but we can still have $f(\mathcal{A}) \geq f(e^*)/(1+\alpha) \geq \text{OPT}/(r(1+\alpha))$ by the assumption in Eq.~\eqref{eq:assumption}. Thus, 
\begin{align*}
    & f(\mathcal{S}) \geq  \ \frac{(1-\epsilon)^{2}(1-2\epsilon)(1-1/\ell)}{p+1+m}\cdot \text{OPT}  \\
    &+\!\frac{(1\!-\!\epsilon)(p\!+\!1)\!-\!(1\!-\!\epsilon)^{3}\!-\!\alpha\left(\epsilon(p\!+\!1)\!+\!m(1\!-\!\epsilon)\right)}{r(1+\alpha)\left( p+1+m(1-\epsilon) \right)} \cdot \text{OPT},
\end{align*}
implying a $(1+\epsilon)(k+m+3+2\sqrt{m+1})$-approximation ratio, since $\alpha\! \leq \!(1\!-\!\epsilon)(p\!+\!1-\!(1\!-\!\epsilon)^{2}) / (\epsilon(p\!+\!1)\!+\!m(1\!-\!\epsilon))$.

We finally estimate the probability of the event $\mathcal{E}_0$ that $\exists \mathcal{A} \!\in \!\{\mathcal{A}_{1},\cdots,\mathcal{A}_{t_c}\}, \mathcal{A} \!\subseteq\! S_{\text{OPT}}$. Let $\mathcal{E}_1$ denote the complement of $\mathcal{E}_0$. By the selection of $\mathcal{A}_{i}$ in line~3 of Algorithm~\ref{algo:SPROUT++}, we have $\text{Pr}[\mathcal{E}_1] = \prod_{i=1}^{t_c}\text{Pr}[\mathcal{A}_i \nsubseteq S_{\text{OPT}}] \leq \prod_{i=1}^{t_c}\left(1-r / (n-i+1)\right) \leq \left(1-r / n\right)^{t_c} \leq e^{- rt_c / n}$, implying $\text{Pr}[\mathcal{E}_0] \!\geq \! 1 \!-\! e^{-rt_c / n}$. Thus, the theorem holds.
\end{proof}

\section{Empirical Study}

In this section, we empirically compare SPROUT++ with celebrated algorithms, i.e., the greedy algorithm, D{\footnotesize ENSITY}S{\footnotesize EARCH}SGS (denoted as DSSGS)~\cite{feldman2020you}, FANTOM~\cite{pmlr-v48-mirzasoleiman16} and RePeated Greedy (denoted as RP\_Greedy)~\cite{DBLP:conf/colt/FeldmanHK17}, on the applications of movie recommendation and weighted max-cut. As SPROUT++ is randomized, we repeat its run for $10$ times independently and report the average and standard deviation. We always the same setting (i.e., $t_c = n / 5$, $\alpha = 0.5$, $\mu = 1$, $\ell = 2$, $\delta = \epsilon = 0.25$, $\beta=5\times 10^{-4}$, and $\gamma=1\times 10^{-6}$) and perform the sensitivity analysis on $t_c$ and $\mu$ to show their influence, and finally compare the performance of SPROUT++ and SPROUT.

\subsubsection{Movie Recommendation.} Movie recommendation is a popular task aiming to pick representative movies. We use the set of $10473$ movies from the MovieLens Dataset, where the rating, release year, genres and feature vector~\cite{lindgren2015sparse}, of each movie are included. To select diverse movies, a non-monotone submodular objective function~\cite{lin2011class, simon2007scene} is considered, i.e., $f(S) = \left(\sum_{i \in \mathcal{N}}\sum_{j \in S} s_{i,j} - \sum_{i \in S}\sum_{j \in S} s_{i,j}\right) / n$, where $s_{i,j} =\exp{(- \lambda \cdot dist(\bm{v}_i,\bm{v}_j))}$ is the similarity between movies $i$ and $j$~\cite{badanidiyuru2020submodular}. $dist(\bm{v}_i,\bm{v}_j)$ denotes the Euclidean distance between movie feature vectors $\bm{v}_i$ and $\bm{v}_j$, and $\lambda$ is set to $4$ in our experiments. As for constraints, we set a uniform matroid for each genre limiting the number of movies in it to $2$, and limit the number of chosen movies to $10$. Moreover, we define three knapsack constraints. The first one leads to movies with higher ratings and its cost function is $c_1 = 10 - rating$ with a budget of $20$. The other two aim to pick movies released close to particular years, i.e., 1995 and 1997, hence the cost functions are $c_2 = |1995 - year|$ and $c_3 = |1997 - year|$ with a budget of $30$.

\subsubsection{Weighted Max-cut.} The max-cut problem is one of the most well-known problems in combinatorial optimization~\citep{edwards1973some}. It aims to partition the vertices of a graph into two sets so that the number of edges between them is maximized. The weighted max-cut problem considers that each edge is associated with a weight and the goal is maximizing the sum of weights of edges in cut. Formally, we study the non-monotone submodular function $f(S) = \sum_{u \in S}\sum_{v \in V \backslash S}w_{u,v}$, where $V$ is the set of vertices and $w_{u,v}$ is the weight of edge $(u, v)$. Concretely, we use a classic type of random graphs, i.e., Erdos Rény graphs~\citep{erdos1960evolution}, where the number of nodes is $n = 1000$, and the probability of each potential edge to be involved in the graph is set to $0.01$. For each edge, we generate its weight uniformly at random from $[0,1]$. As for constraints, we use a uniform matroid to limit the set size to $10$, and set the sum of degree of nodes as its knapsack cost with budget $100$. Moreover, we index the nodes and limit the sum of the last digit of nodes to $40$ as an extra knapsack constraint. 

\begin{figure}[!t]
\centering

\includegraphics[width=0.46\textwidth]{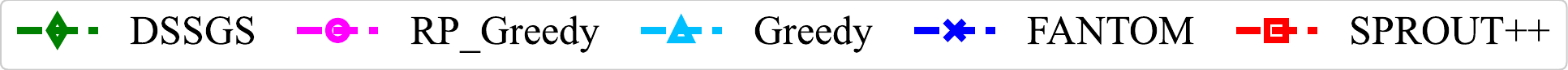}
\label{legend1}
\hfil

\subfloat[Two knapsacks]{\includegraphics[height=0.16\textwidth]{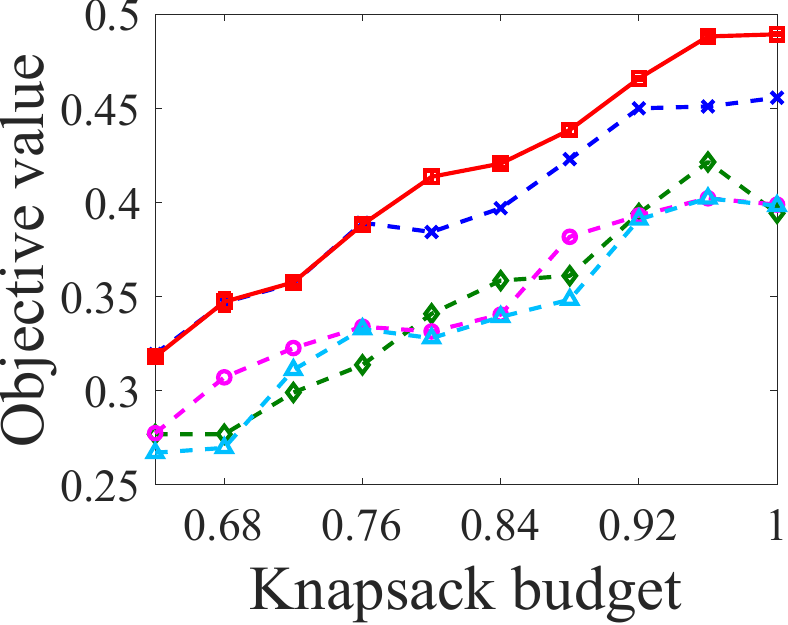}%
\label{a}}
\hfil
\subfloat[Three knapsacks]{\includegraphics[height=0.16\textwidth]{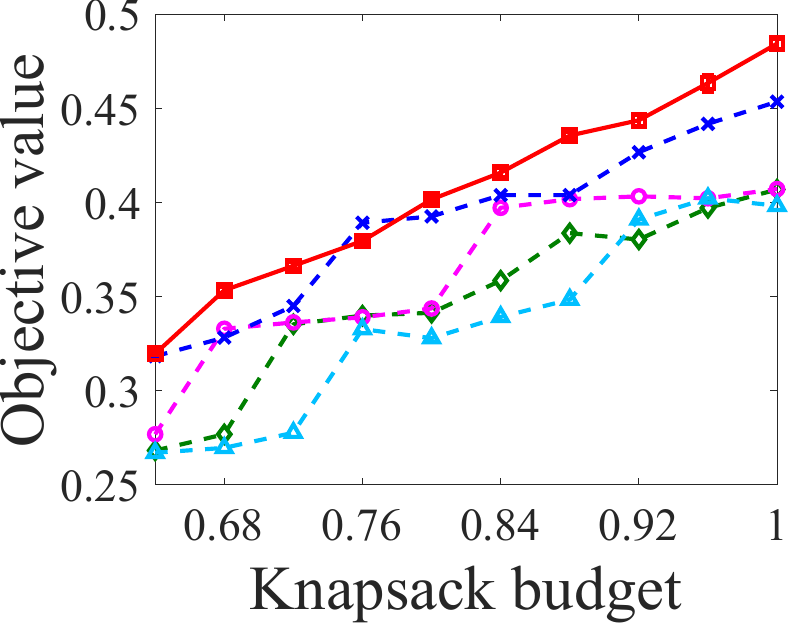}%
\label{b}}
\hfil
\\

\subfloat[Two knapsacks]{\includegraphics[height=0.16\textwidth]{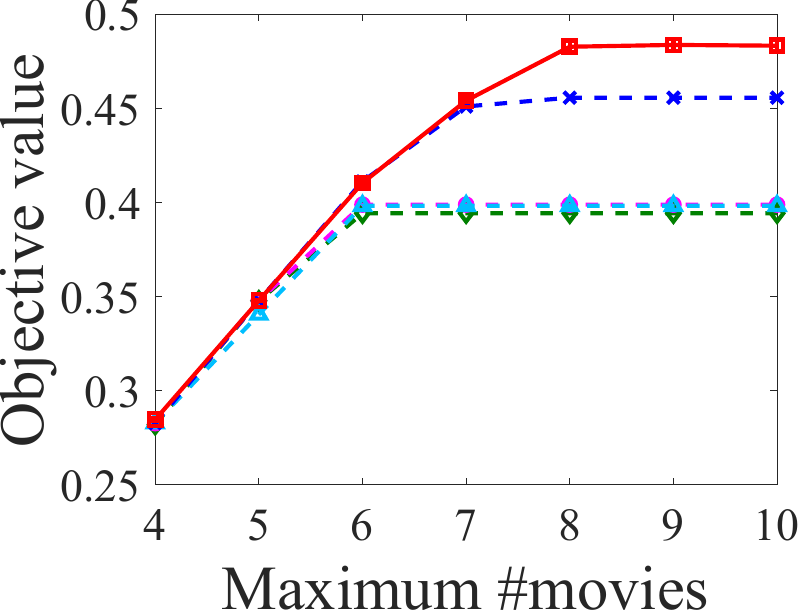}%
\label{c}}
\hfil
\subfloat[Three knapsacks]{\includegraphics[height=0.16\textwidth]{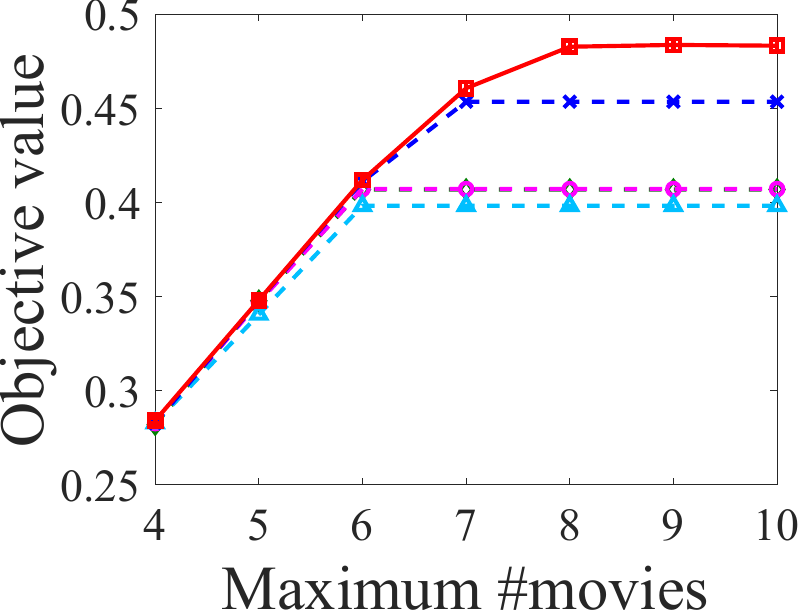}%
\label{d}} 
\hfil 
\caption{Movie Recommendation. (a) and (b): obj. value vs. knapsack budget; (c) and (d): obj. value vs. maximum \#allowed movies.} 
\label{fig1}
\end{figure} 

\begin{figure}[!t]
\centering

\includegraphics[width=0.46\textwidth]{legend.pdf}
\label{legend2}
\hfil

\subfloat{\includegraphics[height=0.155\textwidth]{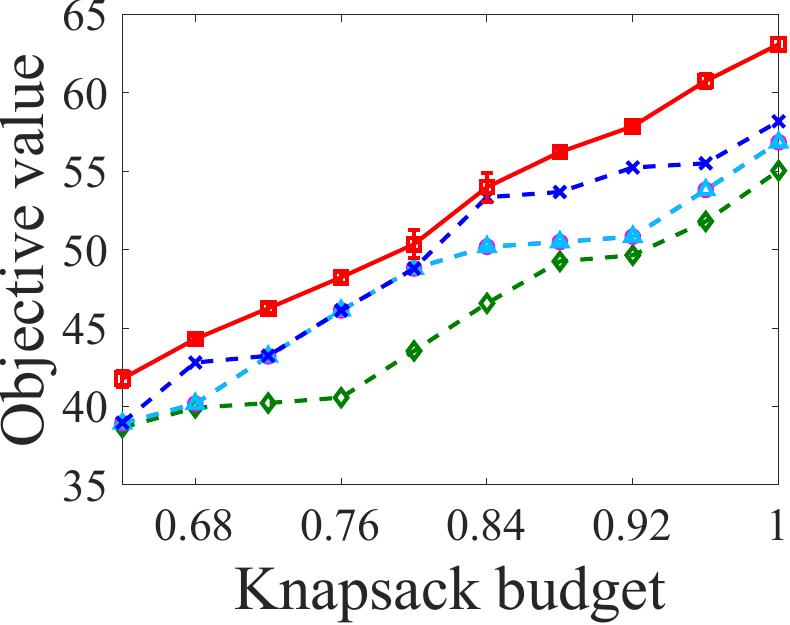}%
\label{e}}
\hfil
\subfloat{\includegraphics[height=0.155\textwidth]{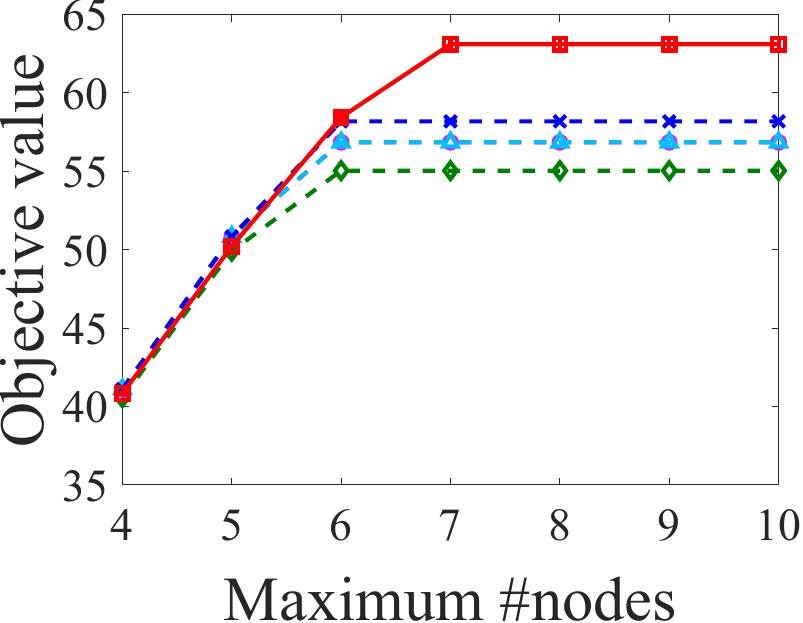}%
\label{f}} 
\hfil

\caption{Weighted Max-cut. Left: obj. value vs. knapsack budget; right: obj. value vs. maximum \#allowed nodes. } 
\label{fig2}
\end{figure} 

The results for experiments are shown in Figures~\ref{fig1} and~\ref{fig2}, respectively. We have normalized the budgets of knapsacks to 1, and then vary them from 0.64 to 1, as shown in
Figures~\ref{fig1}(a),~\ref{fig1}(b) and~\ref{fig2}(a). We also vary the maximum number of allowed movies or nodes from $4$ to $10$ in Figures~\ref{fig1}(c),~\ref{fig1}(d) and~\ref{fig2}(b). For movie recommendation, Figures~\ref{fig1}(a) and~\ref{fig1}(c) consider two knapsacks $c_1$ and $c_2$, while the others consider an extra knapsack constraint $c_3$. These results show that SPROUT++ outperforms all the celebrated algorithms in respect of the quality of chosen movies and graph-cut. Actually, the generally better performance of RP\_Greedy over DSSGS is consistent with the empirical observation in~\citep{feldman2020you}, i.e., RP\_Greedy may be better in practice than DSSGS though the latter algorithm has better worst-case approximation guarantee.

We further study the influence of parameters $t_c$ and $\mu$ empirically. As shown in Figures~\ref{fig3}(a) and ~\ref{fig3}(b), we select FANTOM as the baseline and plot the curve of objective value over the ratio of $t_c$ to $n$ and the value of $\mu$, respectively. Concretely, in Figure~\ref{fig3}(a), we test on the max-cut problem where the budget is set to 1 and the maximum number of allowed nodes is set to 10, and vary the ratio of $t_c$ to $n$ from $0.02$ to $1$. We observe that SPROUT++ immediately overwhelms the runner-up FANTOM; then the objective value returned by SPROUT++ soars with a sharp drop of the standard deviation as the ratio $t_c / n$ increases, and eventually stabilizes at a high level. Meanwhile, we perform experiments on movie recommendation with three knapsack constraints where the budget is set to 1 and the maximum number of allowed movies is set to 10. By setting $t_c=5$ and varying $\mu$ from $1$ to $5$, Figure~\ref{fig3}(b) shows that the performance of SPROUT++ raises as $\mu$ increases, and can catch up with FANTOM, even $t_c$ is extremely small. These results imply that in many cases SPROUT++ can be more efficient and effective in practice than in theoretical analysis.

Finally, we compare the performance of SPROUT++ and SPROUT. We consider the experiment in Figure~\ref{fig2}(a) for examination, and plot their objective value and number of oracle calls in Figure~\ref{fig4}. It comes out that SPROUT++ achieves competitive performance to SPROUT using much less oracle calls, showing the efficiency of SPROUT++ in practice.

\begin{figure}[!t]
\centering
\subfloat{\includegraphics[height=0.16\textwidth]{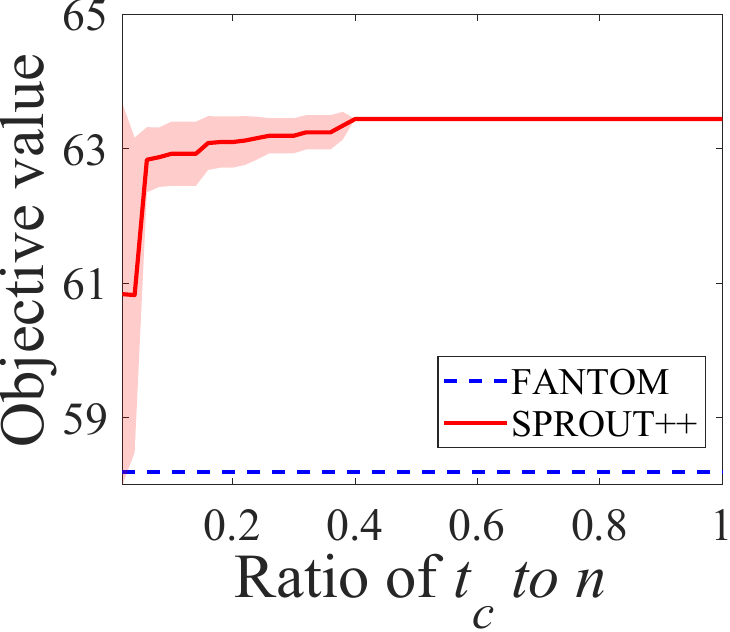}%
\label{g}}
\hfil
\subfloat{\includegraphics[height=0.16\textwidth]{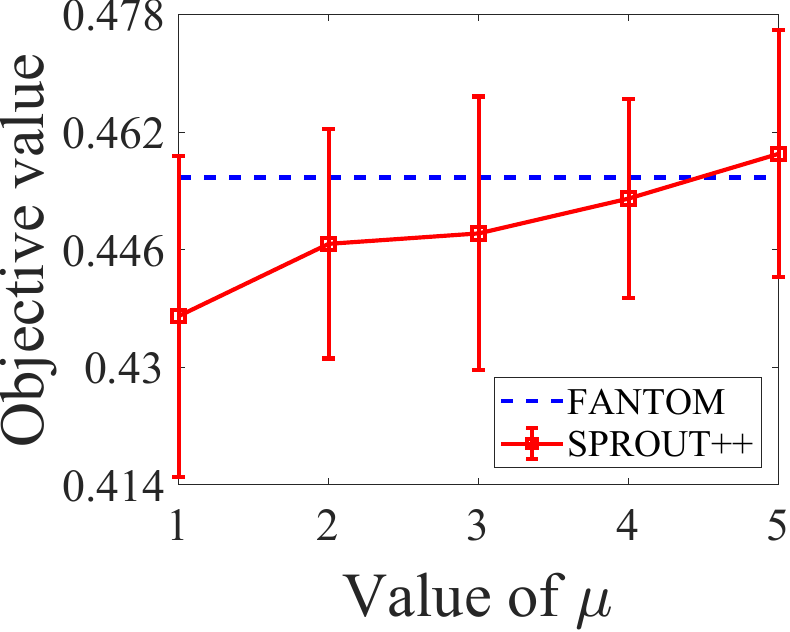}%
\label{h}}
\hfil 
\caption{Parametric sensitivity analysis on $t_c$ and $\mu$. Left: obj. value vs. $t_c / n$ on weighted max-cut; right: obj. value vs. $\mu$ on movie recommendation.}
\label{fig3}
\end{figure}

\begin{figure}[!t]
\centering
\subfloat{\includegraphics[height=0.16\textwidth]{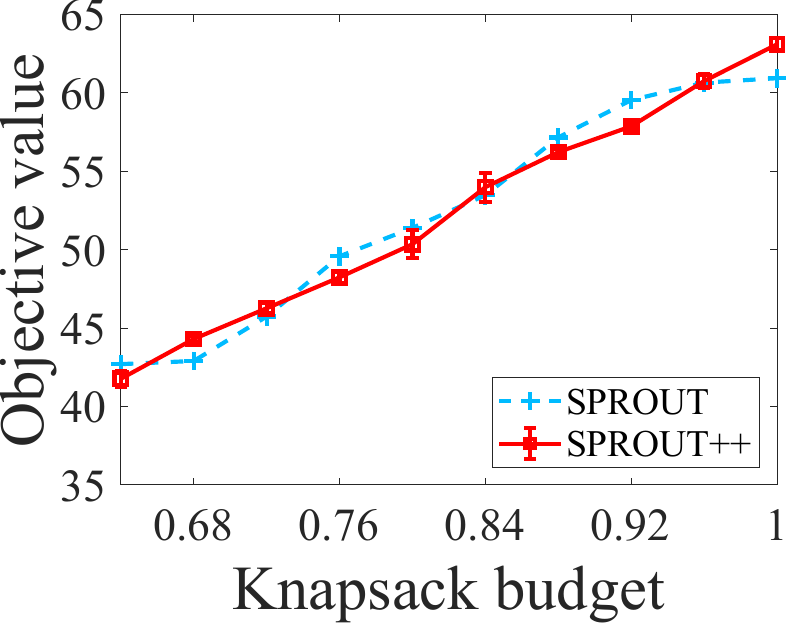}%
\label{i}}
\hfil
\subfloat{\includegraphics[height=0.17\textwidth]{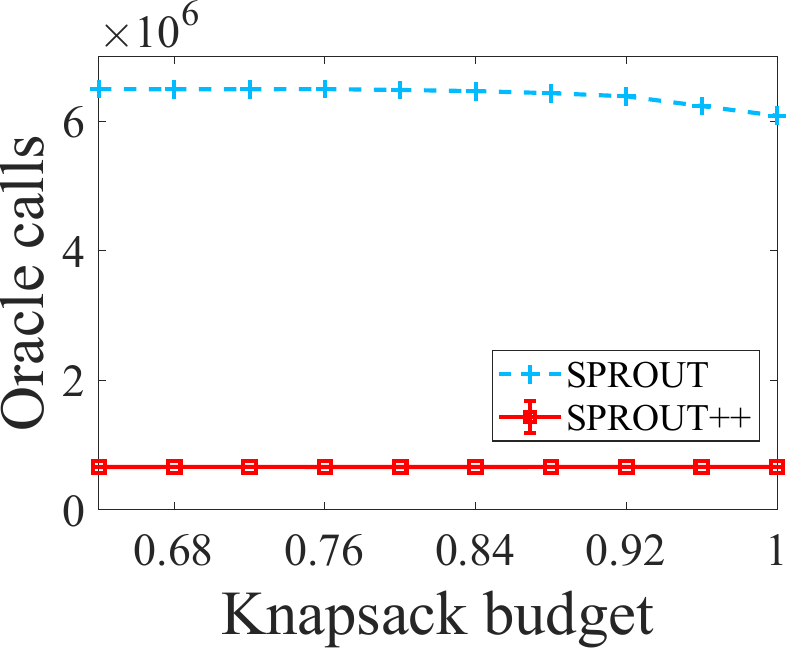}%
\label{j}}
\hfil 
\caption{Comparison of SPROUT/SPROUT++ on experiments in Figure~\ref{fig2}(a). Left: obj. value vs. knapsack budget; right: oracle calls vs. knapsack budget.}
\label{fig4}
\end{figure}

\section{Conclusion}

This paper considers the submodular maximization problem under the intersection of $k$-matroid and $m$-knapsack constraints, and proposes the SPROUT algorithm which achieves a polynomial approximation guarantee better than the best known one. Furthermore, we provide an efficient variant of SPROUT, i.e., SPROUT++, which can still achieve a similar approximation guarantee to SPROUT. Experiments on the applications of movie recommendation and weighted max-cut demonstrate the superiority of SPROUT++ over state-of-the-art algorithms. An interesting future work is to apply our technique to some more general constraints such as $k$-matchoid and $k$-system constraints.


\section{Acknowledgments}

 We want to thank Moran Feldman for the helpful suggestions. This work was supported by the NSFC (62022039, 62276124), the Jiangsu NSF (BK20201247), and the project of HUAWEI-LAMDA Joint Laboratory of Artificial Intelligence. Chao Qian is the corresponding author.

\bibliography{aaai23}

\end{document}